\documentclass[11pt]{amsart}
\usepackage{amsmath}
\usepackage[dvips]{epsfig}
\theoremstyle{plain}
\newtheorem{theorem}{Theorem}[section]
\newtheorem{lemma}[theorem]{Lemma}

\newtheorem{cor}[theorem]{Corollary}

\newtheorem{conjecture}[theorem]{Conjecture}

\theoremstyle{definition}

\numberwithin{equation}{section}

\begin{document}

\title[Lower Bounds for the Area of Black Holes] {Lower Bounds for the Area of Black Holes in Terms of Mass, Charge, and Angular Momentum}

\author[Dain]{Sergio Dain}
\address{Facultad de Matem\'{a}tica, Astronom\'{\i}a y F\'{\i}sica, FaMAF \\
Universidad Nacional de C\'{o}rdoba\\
Instituto de F\'{\i}sica Enrique Gaviola, IFEG, CONICET\\
Ciudad Universitaria (5000) C\'{o}rdoba, Argentina}
\email{dain@famaf.unc.edu.ar}

\author[Khuri]{Marcus Khuri}
\address{Department of Mathematics\\
Stony Brook University\\
Stony Brook, NY 11794, USA}
\email{khuri@math.sunysb.edu}

\author[Weinstein]{Gilbert Weinstein}
\address{Department of Physics; Department of Computer Science and Mathematics\\
Ariel University of Samaria\\
Ariel 40700, Israel}
\email{gilbert.weinstein@gmail.com}

\author[Yamada]{Sumio Yamada}
\address{Department of Mathematics\\
Gakushuin University\\
Tokyo 171-8588, Japan}
\email{yamada@math.gakushuin.ac.jp}

\thanks{The first author acknowledges the support of grant PICT-2010-1387 of CONICET (Argentina) and
grant Secyt-UNC (Argentina). The second author acknowledges the support of
NSF Grant DMS-1007156. The fourth author acknowledges the support of JSPS KAKENHI Grants 23654061 \& 24340009.}
\thanks{PACS numbers: 04.70.Bw, 04.20.Dw, 04.20.Ex}

\begin{abstract}
The most general formulation of Penrose's inequality yields a lower bound for ADM mass in terms of the area,
charge, and angular momentum of black holes.
This inequality is in turn equivalent to an upper and lower bound for the area
in terms of the remaining quantities. In this note, we establish the lower bound
for a single black hole in the setting of axisymmetric maximal initial data sets
for the Einstein-Maxwell equations, when the non-electromagnetic matter fields are
not charged and satisfy the
dominant energy condition. It is shown that the inequality is saturated if and only
if the initial data arise from the extreme Kerr-Newman spacetime. Further refinements
are given when either charge or angular momentum vanish. Lastly, we discuss the validity
of the lower bound in the presence of multiple black holes.
\end{abstract}
\maketitle

\section{Introduction}
\label{sec1} \setcounter{equation}{0}
\setcounter{section}{1}

The standard model of gravitational collapse \cite{Choquet-Bruhat}, \cite{ChruscielGallowayPollack} consists of two
main parts. Namely, gravitational collapse should always result in a black hole (weak cosmic censorship), and
eventually the spacetime should settle down to a stationary electro-vacuum final state. According to the black
hole uniqueness theorem this final state will be a Kerr-Newman spacetime; however it
should be noted that many important technical aspects of black hole uniqueness remain open (see \cite{ChruscielCosta0}
for a recent review). Let $(m_{0},A_{0},q_{0},J_{0})$ denote the mass, black hole area, charge, and angular
momentum of the Kerr-Newman solution, then direct calculation yields an expression for the mass in terms of the remaining quantities
\begin{equation}\label{1}
m_{0}^{2}=\frac{A_{0}}{16\pi}+\frac{q_{0}^2}{2}+\frac{\pi(q_{0}^{4}+4J_{0}^{2})}{A_{0}}.
\end{equation}
In general, the expression on the right-hand side is the square of the so called Christodoulou mass \cite{Christodoulou} of
a black hole. Observe that, as a function of $A_{0}$ (keeping $q_{0}$ and $J_{0}$ fixed), the right-hand
side is nondecreasing precisely when
\begin{equation}\label{2}
A_{0}\geq 4\pi\sqrt{q_{0}^{4}+4J_{0}^{2}}.
\end{equation}

Consider now a Cauchy surface $M$ in an asymptotically flat spacetime which undergoes gravitational collapse,
and settles down to the Kerr-Newman solution above. Let $(m,A,q,J)$ be the ADM mass, black hole area, total charge,
and ADM angular momentum associated with this slice. Since gravitational waves carry positive energy, the ADM (total)
mass of the spacetime should not be smaller than the mass of the final state
\begin{equation}\label{3}
m\geq m_{0},
\end{equation}
and $m-m_{0}$ should measure the amount of radiation emitted by the system. Moreover, the Hawking area theorem
\cite{Hawking} (which relies on cosmic censorship) yields
\begin{equation}\label{4}
A_{0}\geq A.
\end{equation}
Therefore, if conditions are imposed to ensure that the charge and angular momentum are conserved, that is
$q=q_{0}$ and $J=J_{0}$, then we are lead to the following generalized version of the Penrose inequality \cite{Penrose}
\begin{equation}\label{5}
m^{2}\geq\frac{A}{16\pi}+\frac{q^{2}}{2}+\frac{\pi(q^{4}+4J^{2})}{A}\text{ }\text{ }\text{ }\text{ whenever }\text{ }\text{ }
\text{ }A\geq 4\pi\sqrt{q^{4}+4J^{2}}.
\end{equation}
Typical assumptions which guarantee conserved charge and angular momentum are that the spacetime be electro-vacuum
and axially symmetric. Furthermore it is expected that equality is achieved in the first inequality of \eqref{5},
only if $M$ arises from the Kerr-Newman spacetime.

In the case that the area-charge-angular momentum inequality of \eqref{2} is not satisfied, similar arguments motivate
the inequality
\begin{equation}\label{6}
m^2\geq\frac{q^{2}+\sqrt{q^{4}+4J^{2}}}{2}.
\end{equation}
Notice that the right-hand side of \eqref{6} is the minimum value of the right-hand side of \eqref{5}, as a function of $A$, for $A\geq 4\pi\sqrt{q^{4}+4J^{2}}$.
Equality in \eqref{6} should hold only when $M$ arises from the extreme Kerr-Newman spacetime. These two inequalities
yield variational characterizations of the Kerr-Newman and extreme Kerr-Newman initial data, as those with minimal mass
for fixed black hole area, total charge, and angular momentum, or fixed total charge and angular momentum. A violation
of \eqref{5} or \eqref{6} would present a serious challenge to the standard picture of gravitational collapse mentioned above.

The area $A$ appearing in \eqref{5} is that of the intersection of the event horizon with the Cauchy surface $M$. Unfortunately,
from an initial data perspective this is not a useful quantity, since it requires the full evolution
of spacetime in order to determine its value. Thus it is convenient to replace event horizon with apparent horizon, and to
replace $A$ with $A_{min}$, the minimal area required to enclose the outermost apparent horizon or the minimal area required
to enclose all but one asymptotic end.  We may now give the Penrose inequality a purely initial data formulation
\begin{equation}\label{7}
m^{2}\geq\frac{A_{min}}{16\pi}+\frac{q^{2}}{2}+\frac{\pi(q^{4}+4J^{2})}{A_{min}}\text{ }\text{ }\text{ }\text{ whenever }
\text{ }\text{ }\text{ }A_{min}\geq 4\pi\sqrt{q^{4}+4J^{2}}.
\end{equation}
Cosmic censorship implies that the outermost apparent horizon is contained within the event horizon, so that $A\geq A_{min}$
and hence \eqref{7} is implied by \eqref{5}. It follows that a counterexample to \eqref{7} would
be just as significant as a counterexample to \eqref{5} or \eqref{6}.

In this paper we will prove `one half' of inequality \eqref{7} for a single component black hole, under the assumption of maximal
initial data. In order to explain what is meant by `one half', let us multiply
the inequality by $A_{min}$ and view it as a bound for a quadratic polynomial in $A_{min}$. This is then equivalent to the
following upper and lower bound for $A_{min}$
\begin{equation}\label{8}
m^2-\frac{q^2}{2}-\sqrt{\left(m^2-\frac{q^2}{2}\right)^2-\frac{q^4}{4}-J^2}\leq\frac{A_{min}}{8\pi}
\leq m^2-\frac{q^2}{2}+\sqrt{\left(m^2-\frac{q^2}{2}\right)^2-\frac{q^4}{4}-J^2}.
\end{equation}
Notice that the quantity inside the square root is nonnegative by \eqref{6}. We may then state two conjectured inequalities
which are motivated by cosmic censorship
\begin{equation}\label{9}
\frac{A_{min}}{8\pi}\geq m^2-\frac{q^2}{2}-\sqrt{\left(m^2-\frac{q^2}{2}\right)^2-\frac{q^4}{4}-J^2}\text{ }\text{ }\text{ }
\text{ whenever }\text{ }\text{ }\text{ }A_{min}\geq 4\pi\sqrt{q^{4}+4J^{2}},
\end{equation}
and
\begin{equation}\label{10}
\frac{A_{min}}{8\pi}
\leq m^2-\frac{q^2}{2}+\sqrt{\left(m^2-\frac{q^2}{2}\right)^2-\frac{q^4}{4}-J^2}.
\end{equation}
The auxiliary area-charge-angular momentum inequality does not appear with the upper bound \eqref{10} for a reason. Namely,
one may derive this inequality directly from the same type of heuristic arguments which lead to \eqref{7}. In fact, inequality
\eqref{10} is the usual form of the generalized Penrose inequality that appears in the literature \cite{Dain}, \cite{Penrose1}. For maximal initial
data, this upper bound has been proven when $q=J=0$ for a single black hole in \cite{HuiskenIlmanen} and for multiple black holes
in \cite{Bray}; it has also been established in the case $q\neq 0$, $J=0$ for a single black hole \cite{HuiskenIlmanen}, \cite{Jang}
(see \cite{DisconziKhuri} for the case of equality).  The case of equality in \eqref{9} and \eqref{10} should only be achieved by
extreme Kerr-Newman and Kerr-Newman initial data, respectively. The `half' of \eqref{7} which will be established here is the lower
bound \eqref{9}, when the horizon is connected, and without the assumption of the auxiliary area-charge-angular momentum inequality.
In fact when the horizon is connected, the area-charge-angular momentum inequality is a theorem itself, rather than a hypothesis. The
case when $q\neq 0$ and $J=0$, where the hypothesis of axial symmetry is not required, was established in \cite{KhuriWeinsteinYamada}.

\section{Statement and Proof of the Main Result}
\label{sec2} \setcounter{equation}{0}
\setcounter{section}{2}

We begin with the appropriate definitions. Let $(M, g, k, E, B)$ be an initial data set for the Einstein-Maxwell equations, consisting
of a 3-manifold $M$, Riemannian metric $g$, symmetric 2-tensor $k$ (representing the
extrinsic curvature in spacetime), and vector fields $E$ and $B$ representing the electric and magnetic fields. It is assumed that
there is no charged matter, that is
\begin{equation}\label{10'}
\operatorname{div} E=0, \qquad \operatorname{div} B=0.
\end{equation}
Consider a manifold $M$ that has at least two ends, with one being asymptotically flat, and the remainder being
either asymptotically flat or asymptotically cylindrical. Recall that a domain $M_{\text{end}}\subset M$ is an
asymptotically flat end if it
is diffeomorphic to $\mathbb{R}^{3}\setminus\text{Ball}$, and in the coordinates given by the asymptotic
diffeomorphism the following fall-off conditions hold
\begin{equation}\label{11}
g_{ij}=\delta_{ij}+o_{l}(r^{-1/2}),\text{ }\text{ }\text{ }\text{ }\partial g_{ij}\in L^{2}(M_{\text{end}}),\text{
}\text{ }\text{ }
\text{ }k_{ij}=O_{l-1}(r^{-3}),
\end{equation}
\begin{equation}\label{12}
E^{i}=O_{l-1}(r^{-2}),\text{ }\text{ }\text{ }\text{ }\text{ }B^{i}=O_{l-1}(r^{-2}),
\end{equation}
for some $l\geq 6$\footnote{The notation $f=o_{l}(r^{-\alpha})$ asserts that $\lim_{r\rightarrow\infty}r^{\alpha+n}\partial^{n}f=0$
for all $n\leq l$, and
$f=O_{l}(r^{-\alpha})$ asserts that $r^{\alpha+n}|\partial^{n}f|\leq C$ for all $n\leq l$. The assumption $l\geq 6$ is needed for
the results in \cite{Chrusciel}.}. These asymptotics may be weakened, see for example \cite{ChruscielCosta}, \cite{Costa}, and \cite{SchoenZhou}.
If $M$ is simply connected and the data are axially symmetric, it is shown in \cite{Chrusciel} that the analysis reduces to the study
of manifolds of the form
$M\simeq\mathbb{R}^{3}\setminus\sum_{n=0}^{N}i_{n}$, where $i_{n}$ are points in $\mathbb{R}^{3}$ and represent asymptotic ends
(in total there are $N+1$ ends). Moreover there exists a global (cylindrical) Brill coordinate system on $M$, where the points $i_{n}$ all lie on the $z$-axis, and in which the
appropriate asymptotics for the metric coefficients near a cylindrical end are given in \cite{Dain0}. The fall-off conditions in
the asymptotically flat ends guarantee that the asymptotic limits defining the ADM mass, and total charge exist
\begin{equation}\label{13}
m=\frac{1}{16\pi}\int_{S_{\infty}}(g_{ij,i}-g_{ii,j})\nu^{j},
\end{equation}
\begin{equation}\label{14}
q_{e} = \frac1{4\pi} \int_{S_\infty} E_i \nu^i\, , \qquad
q_{b} = \frac1{4\pi} \int_{S_\infty} B_i \nu^i\, ,
\end{equation}
where $S_{\infty}$ indicates the limit as $r\rightarrow\infty$ of integrals over coordinate spheres $S_{r}$, with unit outer
normal $\nu$. Here $q_{e}$ and $q_{b}$ denote the total electric and magnetic charge, respectively,
and we denote the square of the total charge by $q^{2}=q_{e}^{2}+q_{b}^{2}$.

We say that the initial data are axially symmetric if the group of isometries of the Riemannian manifold $(M,g)$ has a subgroup
isomorphic to $U(1)$, and that the remaining quantities defining
the initial data are invariant under the $U(1)$ action. In particular, if $\eta$ denotes the Killing field associated with this
symmetry, then
\begin{equation}\label{15}
\mathfrak{L}_{\eta}g=\mathfrak{L}_{\eta}k=\mathfrak{L}_{\eta}E=\mathfrak{L}_{\eta}B=0,
\end{equation}
where $\mathfrak{L}_{\eta}$ denotes Lie differentiation.
The (gravitational) angular momentum, in the direction of the axis of rotation, of a 2-surface $\Sigma\subset M$ whose tangent space includes $\eta$, is represented by a scalar \cite{Dain}
\begin{equation}\label{16}
J(\Sigma)=\frac{1}{8\pi}\int_{\Sigma}(k_{ij}-(\operatorname{Tr} k)g_{ij})\nu^{i}\eta^{j}.
\end{equation}
The ADM angular momentum is then given by
\begin{equation}\label{17}
J=\lim_{r\rightarrow\infty}J(S_{r}),
\end{equation}
and note that the fall-off conditions \eqref{11} guarantee that the limit exists, as $|\eta|$ grows like $\rho$, the radial
coordinate in the (cylindrical) Brill coordinate system.

Unfortunately, the angular momentum \eqref{16} is not necessarily conserved. We are thus motivated to define an alternate
angular momentum which has this property. In order to do this, we will make use of a vector potential $B=\nabla\times\vec{A}$. However, note that
the topology of $M$ does not allow for a globally defined (smooth) vector potential. The typical construction which avoids this difficulty involves
removing a `Dirac string' associated with each point $i_n$. That is, removing from $M$ either the portion of the $z$-axis below or above $i_{n}$, to obtain
a ($U(1)$ invariant) potential $\vec{A}_{\pm}^{n}$, defined on the complement of the respective Dirac string. We then define
\begin{equation}\label{17.1}
\vec{A}=\frac{1}{2N}\sum_{n=1}^{N}(\vec{A}_{+}^{n}+\vec{A}_{-}^{n})\text{ }\text{ }\text{ }\text{ on }\text{ }\text{ }\text{ }\mathbb{R}^{3}\setminus\{z-\text{axis}\}.
\end{equation}
The (total) angular momentum of a surface $\Sigma$, after the contribution
of the electromagnetic field has been added, is given by
\begin{equation}\label{18}
\widetilde{J}(\Sigma)=\frac{1}{8\pi}\int_{\Sigma}(k_{ij}-(\operatorname{Tr} k)g_{ij})\nu^{i}\eta^{j}+
\frac{1}{4\pi}\int_{\Sigma}(E_i \nu^i)(\vec{A}_{j}\eta^{j}).
\end{equation}
Although $\vec{A}$ is discontinuous on the $z$-axis, the product $\vec{A}_{j}\eta^{j}$ remains well behaved since $\eta$ vanishes on the $z$-axis.
Below it will be shown that this angular momentum is gauge invariant with respect to gauge transformations which vanish in the black hole region,
and is conserved under appropriate conditions on (linear) momentum density.
Recall that the matter density and (linear) momentum density for the non-electromagnetic matter fields are given by the constraint equations
\begin{align}\label{19}
\begin{split}
16\pi \mu  & = R + (\operatorname{Tr} k)^2 - |k|^2 - 2(|E|^2+|B|^{2}), \\
8\pi P & = \operatorname{div} (k - (\operatorname{Tr} k)g)+2E\times B,
\end{split}
\end{align}
where $R$ denotes the scalar curvature of $g$. The non-electromagnetic matter fields will be said to satisfy the dominant
energy condition if
\begin{equation}\label{20}
\mu \geq |P|.
\end{equation}

\begin{lemma}\label{lemma1}
Let $(M,g,k,E,B)$ have the properties described above. If $P_{i}\eta^{i}=0$, then $\widetilde{J}$ is conserved. That is, if
surfaces $\Sigma_{1}$ and $\Sigma_{2}$ are $U(1)$ invariant and bound a domain, then
\begin{equation}\label{21}
\widetilde{J}(\Sigma_{1})=\widetilde{J}(\Sigma_{2}).
\end{equation}
Moreover, $\widetilde{J}$ is invariant under gauge transformations which vanish in a
neighborhood of the points $i_{n}$, and
\begin{equation}\label{22}
\widetilde{J}(S_{\infty})=J.
\end{equation}
\end{lemma}

\begin{proof}
Consider the linear momentum density of the electromagnetic field in the $\eta$-direction
\begin{align}\label{23}
\begin{split}
(E\times B)_{l}\eta^{l}&=\epsilon_{ijl}E^{i}\epsilon^{abj}(\nabla_{a}\vec{A}_{b})\eta^{l}\\
&=\nabla_{a}(\epsilon_{ijl}E^{i}\epsilon^{abj}\vec{A}_{b}\eta^{l})
-\epsilon_{ijl}\epsilon^{abj}\vec{A}_{b}(\nabla_{a}E^{i})\eta^{l}
-\epsilon_{ijl}\epsilon^{abj}\vec{A}_{b}E^{i}\nabla_{a}\eta^{l}\\
&=\nabla_{a}(\epsilon_{ijl}E^{i}\epsilon^{abj}\vec{A}_{b}\eta^{l})+4\vec{A}_{i}\eta^{i}\nabla_{j}E^{j}.
\end{split}
\end{align}
In this calculation we used the fact that $g$ and $E$ are invariant under the $U(1)$ action
\begin{equation}\label{24}
0=(\mathfrak{L}_{\eta}g)_{ij}=\nabla_{i}\eta_{j}+\nabla_{j}\eta_{i}, \qquad 0=(\mathfrak{L}_{\eta}E)^{i}
=\eta^{j}\nabla_{j}E^{i}-E^{j}\nabla_{j}\eta^{i}.
\end{equation}
If $D\subset M$ is the domain with boundary $\partial D=\Sigma_{1}\cup\Sigma_{2}$ then
\begin{equation}\label{25}
\int_{D}E\times B\cdot\eta=\int_{D}4(\vec{A}\cdot\eta) \operatorname{div} E+\int_{\partial D}(\vec{A}\cdot\eta) (E\cdot\nu),
\end{equation}
where in the calculation of the boundary term $\eta\perp\nu$ is used. Therefore, since $\operatorname{div} E=0$
\begin{align}\label{26}
\begin{split}
0=\int_{D}P_{i}\eta^{i}&=\frac{1}{8\pi}\int_{D}\operatorname{div}(k-(\operatorname{Tr} k)g)_{i}\eta^{i}+
\frac{1}{4\pi}\int_{D}(E\times B)_{i}\eta^{i}\\
&=\frac{1}{8\pi}\int_{\partial D}(k_{ij}-(\operatorname{Tr} k)g_{ij})\eta^{i}\nu^{j}+
\frac{1}{4\pi}\int_{\partial D}(E_i \nu^i)(\vec{A}_{j}\eta^{j})\\
&=\widetilde{J}(\Sigma_{1})-\widetilde{J}(\Sigma_{2}).
\end{split}
\end{align}

To show that $\widetilde{J}$ is gauge invariant, consider a gauge transformation $\vec{A}\mapsto \vec{A}+\nabla u$ in which $u$ vanishes near the points $i_n$.
According to \eqref{25} and $\operatorname{div}E=0$, the definition \eqref{18} may be expressed in terms of quantities independent of $u$, as the relevant boundary
integral near $i_n$, and involving $u$, vanishes. Note also that the restriction to gauge transformations vanishing near $i_n$ is physically relevant, since that
restriction is confined within the black hole.

In order to prove \eqref{22}, it suffices to show that
\begin{equation}\label{27}
\lim_{r\rightarrow\infty}\int_{S_{r}}(E_{i} \nu^{i})(\vec{A}_{j}\eta^{j})=0.
\end{equation}
In light of \eqref{12}, we find that $|\vec{A}|=O(r^{-1})$. Moreover $|\eta|\sim|x\partial_y-y\partial_x|=O(\rho)$, where $\rho\sim\sqrt{x^2+y^2}$ is
the radial coordinate in the (cylindrical) Brill coordinate system. Thus, the expansion
\begin{equation}\label{27.1}
E=\frac{q_e}{r^{2}}\partial_{r}+O\left(\frac{1}{r^{3}}\right)
\end{equation}
yields
\begin{equation}\label{27.2}
\lim_{r\rightarrow\infty}\int_{S_{r}}(E_{i} \nu^{i})(\vec{A}_{j}\eta^{j})=\lim_{r\rightarrow\infty}\frac{q_{e}}{r^{2}}\int_{S_{r}}\vec{A}_{j}\eta^{j}.
\end{equation}
Suppose that $i_0$ is situated at the origin. Then since
\begin{equation}\label{27.3}
B=\frac{q_b}{r^{2}}\partial_{r}+O\left(\frac{1}{r^{3}}\right),
\end{equation}
we have
\begin{equation}\label{27.4}
\vec{A}_{+}^{0}=\frac{q_b}{r(z+r)}(x\partial_{y}-y\partial_{x})+O\left(\frac{1}{r^{2}}\right)\text{ }\text{ }\text{ }\text{ on }\text{ }\text{ }\text{ }
\mathbb{R}^{3}\setminus\{(0,0,z)\mid z\leq 0\},
\end{equation}
and
\begin{equation}\label{27.4}
\vec{A}_{-}^{0}=\frac{q_b}{r(z-r)}(x\partial_{y}-y\partial_{x})+O\left(\frac{1}{r^{2}}\right)\text{ }\text{ }\text{ }\text{ on }\text{ }\text{ }\text{ }
\mathbb{R}^{3}\setminus\{(0,0,z)\mid z\geq 0\}.
\end{equation}
Thus
\begin{equation}\label{27.5}
\vec{A}_{+}^{0}\cdot\eta=\frac{q_{b}(x^{2}+y^{2})}{r(z+r)}+O\left(\frac{1}{r}\right)=q_{b}\left(1-\frac{z}{r}\right)+O\left(\frac{1}{r}\right),
\end{equation}
\begin{equation}\label{27.5}
\vec{A}_{-}^{0}\cdot\eta=\frac{q_{b}(x^{2}+y^{2})}{r(z-r)}+O\left(\frac{1}{r}\right)=-q_{b}\left(1+\frac{z}{r}\right)+O\left(\frac{1}{r}\right),
\end{equation}
and it follows that
\begin{equation}\label{28}
\lim_{r\rightarrow\infty}\frac{1}{r^{2}}\int_{S_{r}}(\vec{A}_{+}^{0}+\vec{A}_{-}^{0})\cdot\eta=0.
\end{equation}
Similarly, this type of cancelation occurs for the pair of potentials associated with each $i_{n}$. The desired result \eqref{27} now follows from \eqref{27.2}
\end{proof}

In order to establish \eqref{9}, we require the global inequality \eqref{6}, which relies on the existence of a twist potential $\omega$:
\begin{equation}\label{29}
\epsilon_{ijl}(\pi^{jn}+2\chi^{jn})\eta^{l}\eta_{n}dx^{i}=d\omega
\end{equation}
where
\begin{equation}\label{29.1}
\pi_{jn}=k_{jn}-(\operatorname{Tr}k)g_{jn}, \qquad  \chi_{jn}=\epsilon_{imn}E^{i}\epsilon_{j}^{\text{ }\!\text{ }lm}\vec{A}_{l}.
\end{equation}
Such a potential exists, for example, in the electro-vacuum setting \cite{Weinstein}. Here we show that a weaker condition is sufficient.

\begin{lemma}\label{lemma2}
Let $(M,g,k,E,B)$ have the properties described above. If $P_{i}\eta^{i}=0$, then a twist potential exists.
\end{lemma}

\begin{proof}
For any 2-tensor $t_{ij}$ (not necessarily symmetric), consider the expression
\begin{equation}\label{29.2}
\mathcal{T}_{i}=\epsilon_{ijl}t^{jn}\eta^{l}\eta_{n}.
\end{equation}
A direct calculation shows that
\begin{equation}\label{29.3}
(d\mathcal{T})_{ij}=\nabla_{i}\mathcal{T}_{j}-\nabla_{j}\mathcal{T}_{i}=\nabla^{a}(t_{ab}\eta^{b})\epsilon_{jil}\eta^{l}.
\end{equation}
Thus if we choose
\begin{equation}\label{29.4}
t_{ij}=\pi_{ij}+2\chi_{ij},
\end{equation}
then
\begin{align}\label{29.5}
\begin{split}
(d\mathcal{T})_{ij}&=\nabla^{a}(\pi_{ab}\eta^{b}+2\chi_{ab}\eta^{b})\epsilon_{jil}\eta^{l}\\
&=[(\nabla^{a}\pi_{ab})\eta^{b}+2\nabla^{a}(\chi_{ab}\eta^{b})]\epsilon_{jil}\eta^{l}\\
&=[(\nabla^{a}\pi_{ab})\eta^{b}+2(E\times B)_{b}\eta^{b}-8\vec{A}_{b}\eta^{b}\nabla_{a}E^{a}]\epsilon_{jil}\eta^{l}\\
&=(8\pi P_{b}\eta^{b}-8\vec{A}_{b}\eta^{b}\nabla_{a}E^{a})\epsilon_{jil}\eta^{l},
\end{split}
\end{align}
where we have used that $\pi_{ij}$ is symmetric and $\eta^i$ is a Killing field, as well as formula \eqref{23}. Therefore since $P_{i}\eta^{i}=\operatorname{div}E=0$, $\mathcal{T}$ is
a closed 1-form when $t$ is given by \eqref{29.4}. As $M$ is simply connected, it follows that a twist potential exists.
\end{proof}

The precise statement of conditions under which the inequality \eqref{6} is valid, will now be reviewed. Typically such a result is stated in the
electrovacuum setting, since this guarantees the existence of a twist potential. However with Lemma \ref{lemma2}, we obtain a slight generalization by replacing
the electrovacuum assumption with the dominant energy condition and $P_{i}\eta^{i}=0$. Recall also that the initial
data are said to be maximal if $\operatorname{Tr}k=0$.

\begin{theorem}[\cite{ChruscielCosta,Costa,SchoenZhou}]\label{thm1}
Let $(M,g,k,E,B)$ be a simply connected, axially symmetric, maximal initial data set with two ends, one asymptotically
flat and the other either asymptotically flat or asymptotically cylindrical.
If there is no charged matter, the dominant energy condition is satisfied, and $P_{i}\eta^{i}=0$ (all of which
are satisfied in electrovacuum), then
\begin{equation}\label{30}
m^2\geq\frac{q^{2}+\sqrt{q^{4}+4J^{2}}}{2}.
\end{equation}
Moreover, equality holds if and only if the initial data arise from an extreme Kerr-Newman spacetime.
\end{theorem}

We will also make use of the area-charge-angular momentum inequality. In the case where both ends of $M$ are asymptotically
flat, there exists an axisymmetric stable minimal surface $\Sigma_{min}\subset M$, separating
both ends. Here, stable means that the second variation of area is nonnegative. $\Sigma_{min}$ minimizes area among all
closed 2-surfaces enclosing either end, and thus $A(\Sigma_{min})=A_{min}$,
where $A_{min}$ is the least area required to enclose an end. If one of the ends is cylindrical, the least area $A_{min}$
required to enclose this end is realized either by a stable minimal surface
$\Sigma_{min}\subset M$, or by the surface $\Sigma_{0}$ obtained by taking the asymptotic limit of cross-sections
$\Sigma_{\rho}$ of the end. These observations allow an application of the results in \cite{Reiris} to
obtain the next theorem. For any 2-surface $\Sigma\subset M$, define $q(\Sigma)^{2}=q_{e}(\Sigma)^{2}+q_{b}(\Sigma)^{2}$
where $q_{e}(\Sigma)$ and $q_{b}(\Sigma)$ are defined analogously to \eqref{14} with the
only difference being that the integrals are taken over $\Sigma$. Define also
\begin{equation}\label{31}
q(\Sigma_{0})=\lim_{\rho\rightarrow 0}q(\Sigma_{\rho}), \qquad \widetilde{J}(\Sigma_{0})
=\lim_{\rho\rightarrow 0}\widetilde{J}(\Sigma_{\rho}).
\end{equation}

\begin{theorem}[\cite{Reiris}]\label{thm2}
Let $(M,g,k,E,B)$ be an axially symmetric, maximal initial data set with two ends, one asymptotically flat and the other
either asymptotically flat or asymptotically cylindrical.
If the dominant energy condition is satisfied, then
\begin{equation}\label{32}
A_{min}\geq 4\pi\sqrt{q(\Sigma)^{4}+4\widetilde{J}(\Sigma)^{2}}
\end{equation}
where $\Sigma$ denotes either $\Sigma_{min}$ or $\Sigma_{0}$. Moreover, equality is achieved if and only if $\Sigma=\Sigma_{0}$
is the extreme Kerr-Newman sphere.
\end{theorem}

We may now state and prove our main result.

\begin{theorem}\label{thm3}
Let $(M,g,k,E,B)$ be a simply connected, axially symmetric, maximal initial data set with two ends, one $(M_{\text{end}}^{1})$
asymptotically flat and the other $(M_{\text{end}}^{2})$ either asymptotically flat or asymptotically cylindrical. If there is
no charged matter, the dominant energy condition is satisfied, and $P_{i}\eta^{i}=0$ (all of which are
satisfied in electrovacuum), then
\begin{equation}\label{33}
\frac{A_{min}}{8\pi}\geq m^2-\frac{q^2}{2}-\sqrt{\left(m^2-\frac{q^2}{2}\right)^2-\frac{q^4}{4}-J^2},
\end{equation}
where $A_{min}$ is the minimum area required to enclose $M_{\text{end}}^{2}$.
Moreover, equality holds if and only if the initial data arise from an extreme Kerr-Newman spacetime.
\end{theorem}

\begin{proof}
Apply Lemma \ref{lemma1} and Theorems \ref{thm1} and \ref{thm2} to find
\begin{align}\label{34}
\begin{split}
m^{2}-\frac{q^{2}}{2}&=\sqrt{\left(m^2-\frac{q^2}{2}\right)^2-\frac{q^4}{4}-J^2+\frac{q^{4}}{4}+J^{2}}\\
&\leq\sqrt{\frac{q^{4}}{4}+J^{2}}+\sqrt{\left(m^2-\frac{q^2}{2}\right)^2-\frac{q^4}{4}-J^2}\\
&=\sqrt{\frac{q(\Sigma)^{4}}{4}+\widetilde{J}(\Sigma)^{2}}+\sqrt{\left(m^2-\frac{q^2}{2}\right)^2-\frac{q^4}{4}-J^2}\\
&\leq\frac{A_{min}}{8\pi}+\sqrt{\left(m^2-\frac{q^2}{2}\right)^2-\frac{q^4}{4}-J^2},
\end{split}
\end{align}
where $\Sigma$ is as in Theorem \ref{thm2}.

In the case of equality, we must have either
\begin{equation}\label{35}
\frac{q^{4}}{4}+J^{2}=0\text{ }\text{ }\text{ }\text{ or }\text{ }\text{ }\text{ }\left(m^2-\frac{q^2}{2}\right)^2-\frac{q^4}{4}-J^2=0.
\end{equation}
If the first equality is satisfied, then $q=J=0$. From equality in \eqref{34} we then obtain $A_{min}=0$, which contradicts the
existence of two ends. Thus, the second equality must hold.
This however is equivalent to the case of equality in \eqref{30}. The desired result now follows from Theorem \ref{thm1}.
\end{proof}

\section{Charge and Angular Momentum Separately}
\label{sec3} \setcounter{equation}{0}
\setcounter{section}{3}

It is typically thought that charge and angular momentum behave in a somewhat similar manner with regard to such geometric
inequalities. However in the context of \eqref{33}, their behavior is
quite different when multiple horizons are present. Let us first consider the case of charge alone, that is $q\neq 0$ and
$J=0$. In this case the assumption of simple connectivity, axial symmetry,
maximality, and the existence of a twist potential are not required.

\begin{theorem}\label{thm4}
Let $(M,g,k,E,B)$ be an initial data set without charged matter and satisfying the dominant energy condition.

1) If the initial data are asymptotically flat with one end, and possess a single component boundary consisting of
an outermost apparent horizon with area $A$, then
\begin{equation}\label{36}
\sqrt{\frac{A}{4\pi}}\geq m-\sqrt{m^2-q^{2}}.
\end{equation}

2) If the initial data are maximal with two ends, one $(M_{\text{end}}^{1})$ asymptotically flat and the other
$(M_{\text{end}}^{2})$ either asymptotically flat or asymptotically cylindrical, then
\begin{equation}\label{37}
\sqrt{\frac{A_{min}}{4\pi}}\geq m-\sqrt{m^2-q^{2}},
\end{equation}
where $A_{min}$ is the minimum area required to enclose $M_{\text{end}}^{2}$.
Moreover, equality holds if and only if the initial data arise from an extreme Reissner-Nordstr\"{o}m spacetime.
\end{theorem}

\begin{proof}
Statement (1) is proven in \cite[Corollary 2]{KhuriWeinsteinYamada}. The inequality in statement (2) is equivalent to
\eqref{33} in Theorem \ref{thm3} when $J=0$, and may be established in the same way, since the positive mass theorem
with charge \cite{CRT}, \cite{GHHP} as well as the area-charge inequality \cite{Reiris}, \cite{DJR} are valid under the
current hypotheses.
\end{proof}

Ideally one would like to show that \eqref{36} holds when $A$ is replaced by the minimum area required to enclose the
outermost apparent horizon. This is of course a stronger result, however the relevant area-charge inequality
needed to establish it is currently not available. Moreover the case of equality in \eqref{36} should also imply that
the initial data arise from the extreme Reissner-Nordstr\"{o}m spacetime. However the relevant case of equality
for the positive mass theorem with charge, needed to establish this result, is also currently not available, although
a proposal for its proof has been put forth in \cite{KhuriWeinstein}.

Consider now the case of angular momentum alone, that is $q=0$ and $J\neq 0$. The situation for a single black hole is
very similar to that of charge alone. For instance, as a corollary of Theorem \ref{thm3} we have the following
statement.

\begin{cor}\label{cor1}
Let $(M,g,k)$ be a simply connected, axially symmetric, maximal initial data set with two ends, one $(M_{\text{end}}^{1})$
asymptotically flat and the other $(M_{\text{end}}^{2})$ either asymptotically flat or asymptotically cylindrical. If the
dominant energy condition is satisfied, and $P_{i}\eta^{i}=0$ (all of which are satisfied in vacuum),
then
\begin{equation}\label{38}
\frac{A_{min}}{8\pi}\geq m^2-\sqrt{m^4-J^2},
\end{equation}
where $A_{min}$ is the minimum area required to enclose $M_{\text{end}}^{2}$.
Moreover, equality holds if and only if the initial data arise from an extreme Kerr spacetime.
\end{cor}

When multiple black holes are present, similarity between the charged case and the angular momentum case break down. To see this,
recall that the Majumdar-Papapetrou spacetime --- the static extremal black hole solution to
the electrovacuum equations --- violates the area-charge inequality whenever two or more black holes are present. Based on this
observation, Weinstein and Yamada were able to perturb Majumdar-Papapetrou initial data to
find the following counterexample to the lower bound \eqref{37}.

\begin{theorem}[\cite{WeinsteinYamada}]\label{thm5}
There exists a time symmetric $(k=0)$, asymptotically flat initial data set $(M,g)$ for the Einstein-Maxwell
system, having outermost minimal surface boundary (with two components) and such that
\begin{equation}\label{39}
\frac{A}{4\pi}< m-\sqrt{m^2-q^{2}},
\end{equation}
where $A$ is the area of the boundary.
\end{theorem}

On the other hand, consider the case of multiple black holes with angular momentum alone. Let us label the areas of the stable
minimal surfaces and angular momentums of each black hole by $A_{i}$ and $J_{i}$, $i=1,\ldots,I$
respectively. Under the assumption of maximal axisymmetric initial data, satisfying the dominant energy condition, the
area-angular momentum inequality \cite{Reiris} for single black holes implies that
\begin{equation}\label{40}
A=\sum_{i=1}^{I}A_{i}\geq \sum_{i=1}^{I}8\pi|J_{i}|\geq 8\pi\left|\sum_{i=1}^{I}J_{i}\right|=8\pi|J|.
\end{equation}
Thus, the area-angular momentum inequality holds for multiple black holes, since the left-hand side is additive in $A$ and
subadditive in $J$.  This leads to the following conjecture.

\begin{conjecture}\label{conj1}
Let $(M,g,k)$ be a simply connected, axially symmetric, maximal initial data set with multiple ends, one $(M_{\text{end}}^{1})$
asymptotically flat and the others $(M_{\text{end}}^{i})$, $i=2,\ldots,I$ either asymptotically flat or asymptotically cylindrical.
If the dominant energy condition is satisfied, and $P_{i}\eta^{i}=0$ (all of which are satisfied in vacuum), then
\begin{equation}\label{41}
\frac{A}{8\pi}\geq m^2-\sqrt{m^4-J^2},
\end{equation}
where $J$ is the sum of the angular momentums, and $A$ is the sum of the areas of the stable minimal surfaces enclosing each end
$M_{\text{end}}^{i}$, $i=2,\ldots,I$.
Moreover, equality holds if and only if the initial data arise from an extreme Kerr spacetime.
\end{conjecture}

If the positive mass theorem with angular momentum for multiple black holes were known to be valid
\begin{equation}\label{42}
m^2\geq\left|\sum_{i=1}^{I}J_{i}\right|,
\end{equation}
then we could establish \eqref{41} with the help of \eqref{40}
\begin{align}\label{43}
\begin{split}
m^2=\sqrt{m^4-J^2 + J^2 }&\leq |J| + \sqrt{m^4-J^2}\\
&\leq \frac{A}{8\pi}+\sqrt{m^4-J^2}.
\end{split}
\end{align}
Furthermore, there is strong physical evidence in support of \eqref{42}. Namely, the same heuristic arguments presented
in Section \ref{sec1}, and based on cosmic censorship, may be used to
derive \eqref{42}. It then appears likely that Conjecture \ref{conj1} is true. Hence we find distinctly different behavior
with regard to the lower bound \eqref{33} in the case of multiple
black holes, as counterexamples exist when $q\neq 0$, $J=0$ and counterexamples should not exist when $q=0$, $J\neq 0$.
The key reason for this difference seems to be that the area-angular momentum inequality is subadditive in $J$, whereas
the area-charge inequality does not have this property for $q$. Moreover, there do not exist analogues of the Majumdar-Papapetrou
solutions in the case of angular momentum alone, and somehow
the inequality \eqref{33} seems to know this fact.

\end{document}